\documentclass[aps,pra,reprint,groupedaddress]{revtex4-1}
\usepackage{amsmath,graphicx,hyperref,amsfonts,amsthm}
\usepackage{indentfirst}
\usepackage{braket}

\usepackage{xcolor}
\usepackage{diagbox}
\usepackage{amssymb}
\usepackage{bbold}
\usepackage{multirow}
\renewcommand{\eqref}[1]{Eq.~\ref{#1}}

\newcommand{\figureref}[1]{Fig. ~\ref{#1}}
\newcommand{\tr}{{\mathrm{Tr}}}
\newcommand{\dd}{{\mathrm{d}}}
\newcommand{\centered}[1]{\begin{tabular}{l} #1 \end{tabular}}

\newtheorem{theorem}{Theorem}
\newtheorem{definition}{Definition}
\newtheorem{Lemma}{Lemma}

\begin{document}
\title{Open quantum systems integrable by partial commutativity}
\author{Artur Czerwinski}
\email{aczerwin@umk.pl}
\affiliation{Institute of Physics, Faculty of Physics, Astronomy and Informatics, Nicolaus Copernicus University, Grudziadzka 5, 87--100 Torun, Poland}

\begin{abstract}
The article provides a framework to solve linear differential equations based on partial commutativity which is introduced by means of the Fedorov theorem. The framework is applied to specific types of three-level and four-level quantum systems. The efficiency of the method is evaluated and discussed. The Fedorov theorem appears to answer the need for methods which allow to study dynamical maps corresponding with time-dependent generators. By applying this method, one can investigate countless examples of dissipative systems such that the relaxation rates depend on time.
\end{abstract}

\maketitle

\section{Introduction}

The problem of solving a differential equation belongs to most fundamental issues in the theory of open quantum systems. The ability to obtain a solution in the closed form means that one can determine the trajectory of the system, which provides complete characterization of how the quantum state changes in time. However, only particular types of differential equations allow solutions in closed forms. Additionally, a universal criterion for integrability does not exist. Therefore, there is a need for new methods which can be applied to investigate different types of equations. In this article we propose to implement the Fedorov theorem in the theory of open quantum systems.

The simplest dynamical map, which does not need any further comment at this point, can be obtained when the time-evolution is given by a master equation with the GKSL generator $\mathbb{L}: \mathbb{M}_N (\mathbb{C}) \rightarrow \mathbb{M}_N (\mathbb{C})$, where we assume that the space is finite-dimensional  \cite{Gorini1976,Lindblad1976,Manzano2020}. In such a case, the density matrix at any time instant can be computed by the semigroup:
\begin{equation}\label{eq34}
\rho (t) = \exp \left(\mathbb{L} t\right)\, [\rho(0)],
\end{equation}
where $\rho(0)$ stands for the initial density matrix. A master equation governed by the GKSL generator is the most general type of Markovian and time-homogeneous evolution which preserves trace and positivity.

The closed-form solution of a master equation can be obtained straightforwardly as long as the generator is time-independent. The problem appears when the dynamics is governed by a master equation with a time-dependent linear generator:
\begin{equation}\label{eq35}
\frac{\dd \rho(t)}{\dd t} = \mathbb{L}(t)\, [\rho(t)],
\end{equation}
where the generator $\mathbb{L}(t)$ is defined on a time interval $\mathcal{I}$.

In 1949, Dyson published an article \cite{Dyson1949}, in which he presented the formal solution of an explicitly time-dependent Schr\"{o}dinger equation. The result was obtained by iteration and a time ordering operator, which was later called after the author "Dyson series". Thus, the formal solution of \eqref{eq35} can be written by means of a superoperator $\Phi(t)$:
\begin{equation}\label{eq26}
\rho (t) = \Phi(t)\, [\rho(0)] = \mathrm{T} \:\exp\left( \int_{0}^t \mathbb{L}(\tau) d \tau\right) [\rho(0)],
\end{equation}
where $\mathrm{T}$ denotes the chronological product. The formula for the map $\Phi(t)$ can be expanded by applying the Dyson series \cite{Dyson1949}:
\begin{equation}\label{eq27}
\Phi(t) = \mathbb{1}_N + \int_{0}^t \dd t_1 \mathbb{L}(t_1) + \int_{0}^t \dd t_1 \int_{0}^{t_1} d t_2 \mathbb{L}(t_1) \mathbb{L}(t_2)+ \dots,
\end{equation}
provided it converges. One fundamental problem studied in the theory of open quantum systems relates to algebraic properties of $\mathbb{L}(t)$ which guarantee that the solution $\Phi(t)$ constitutes a legitimate dynamical map, e.g. \cite{Alicki2007}. Undoubtedly, such a question is relevant, but in this article we focus on the methods which provide solutions to time-dependent master equations of the form \eqref{eq35} without the necessity to utilize the infinite Dyson series.

In Sec. \ref{section1}, we revise the definitions and theorems connected with functionally commutative generators. Then, in Sec. \ref{section2}, we present the Fedorov theorem, which can be understood as a generalization of the Lappo-Danilevsky criterion. Along with the theorem we propose a feasible framework for its application in concrete examples. Then, in Sec. \ref{section3}, the framework is tested as we apply the Fedorov theorem to three-level and four-level open quantum systems with evolution governed by time-local generators. We study three particular types of three-level dynamics: $V-$system, \textit{cascade} and \textit{Lambda}, as well as one example on four-level \textit{cascade} systems, in order to prove that this technique can facilitate solving master equations with time-dependent generators.

\section{Functional and integral commutativity}\label{section1}

In order to follow the trajectory of the systems, it is desirable to be able to write the solution of \eqref{eq35} in the closed form:
\begin{equation}\label{eq36}
\rho(t) = \exp \left( \int_0^t \mathbb{L}(\tau) d \tau \right) [\rho(0)],
\end{equation}
which can be done only for specific generators $\mathbb{L}(t)$.

First, we shall analyze the sufficient conditions which, if satisfied by the generator $\mathbb{L}(t)$, guarantee that the solution can be written in the closed form. We shall refer to algebraic properties of the matrix representation of the generator $\mathbb{L}(t)$.

To begin with, let us recall a definition, assuming that $\textbf{F}(t)$ stands for a matrix function and $\mathcal{I}$ denotes an interval within its domain.
\begin{definition}[Semiproper matrix function]
A matrix function $\textbf{F}: \mathcal{I} \rightarrow \mathbb{C}^{n\times n}$ is called \textit{semiproper} on $\mathcal{I}$ if
\begin{equation}\label{eq37}
\textbf{F}(t) \textbf{F}(\tau) = \textbf{F}(\tau)\textbf{F}(t) \;\; \forall \:t, \tau \in \mathcal{I}.
\end{equation}
\end{definition}

The definition of the semiproper function can be applied to time-dependent generators of evolution, which are a specific kind of complex-valued time-dependent function matrices. In other words, this property is called functional commutativity.

\begin{definition}[Functional commutativity]
A time-dependent generator $\mathbb{L}(t)$ is functionally commutative (i.e. semiproper) iff
\begin{equation}\label{eq39}
[\mathbb{L}(t), \mathbb{L}(s) ] = 0 \;\; \forall\: t,s \in \mathcal{I}.
\end{equation}
\end{definition}

The notion of functional commutativity applied to generators of evolution allows one to formulate a theorem concerning the solvability of the dynamics \eqref{eq35} \cite{Erugin1966,Lukes1982}.

\begin{theorem}\label{solv1}
If the generator of evolution $\mathbb{L}(t)$ satisfies the condition of functional commutativity \eqref{eq39}, then the solution of \eqref{eq35} can be written in the closed form according to \eqref{eq36}.
\end{theorem}

The idea of semiproper matrix functions has received much attention in the second half of the XX century. One noteworthy article was written by J. Martin in 1967 \cite{Martin1967}. In one of the theorems, the author proved that the family of semiproper matrix functions can be completely characterized as commutative algebras generated by a basis of pairwise commutative constant matrices. Based on this result, we can say that $\mathbb{L}(t)$ is functionally commutative on $\mathcal{I}$ iff there exists a set of mutually commuting time-independent matrices $\{\mathbb{L}^{(k)}\}$ and piecewise continuous scalar functions $\{\alpha_k (t)\}$ such that
\begin{equation}\label{martin}
\mathbb{L}(t) = \sum_{k} \alpha_k (t) \mathbb{L}^{(k)}.
\end{equation}

The decomposition of the generator of evolution \eqref{martin} not only allows one to write the closed-form solution of \eqref{eq35}, but also simplifies the computing of the integral over time. However, finding such a decomposition of $\mathbb{L}(t)$ remains a challenge \cite{Zhu1990}. For this reason, J. Zhu proposed a different method to decompose a functionally commutative generator (called the spatial decomposition) \cite{Zhu1992}, which was later developed by T. Kamizawa and applied to open quantum systems \cite{Kamizawa2015}.

Another approach to the problem of solving the evolution equation of the form \eqref{eq35} is to apply to the notion of commutativity with the integral. It is another condition which, if satisfied by the the generator $\mathbb{L}(t)$, implies that the solution of the evolution equation is given in the closed form. Let us recall the definition.

\begin{definition}[Integral commutativity]
A time-dependent generator $\mathbb{L}(t)$ is said to commute with its integral iff:
\begin{equation}\label{eq40}
\mathbb{L}(t) \int \mathbb{L}(t) d t = \int \mathbb{L}(t) d t\; \mathbb{L}(t) \: \iff \: [ \mathbb{L}(t),  \int \mathbb{L}(t) d t] = 0.
\end{equation}
\end{definition}

A thorough study of time-dependent matrices which commute with their integrals was published by Bogdanov and Cheboratev in 1959 \cite{Bogdanov1959}. Necessary and sufficient conditions for integral commutativity can be given in relation to the properties of the Jordan canonical form of $\mathbb{L}(t)$ \cite{Epstein1963,Evard1990}. Based on the notion of integral commutativity, one can formulate a theorem concerning the solvability of evolution equations \cite{Lappo1957}.

\begin{theorem}\label{solv2}
If the generator of evolution $\mathbb{L}(t)$ satisfies the condition of integral commutativity \eqref{eq40}, then the fundamental solution of \eqref{eq35} has the closed form \eqref{eq36}.
\end{theorem}

It is worth noting that $\mathbb{L}(t)$ is said to be analytic in a neighborhood of $t=t_0$ when each element of $\mathbb{L}(t)$ (and thus $\mathbb{L}(t)$ itself) can be represented as a Taylor series centered at $t_0$ which converges in some neighborhood of $t_0$. If the time-dependent generator $\mathbb{L}(t)$ is an analytic complex valued matrix function, then $\mathbb{L}(t)$ satisfies the condition of functional commutativity if and only if it commutes with its integral, which means that in such a case both criteria are compatible \cite{Goff1981}.

The theorem \ref{solv2} could be equivalently formulated in terms of the generator which commutes with its derivative, which a common way to express and study this criterion, e.g. \cite{Evard1985,Turcotte2002,Maouche2020}. Nonetheless, for the sake of the content of this article, we stay with the notion of integral commutativity, originally introduced by Lappo-Danilevsky, which is a starting point for further analysis.

\section{Partial commutativity and a framework for its application}\label{section2}

Either functional or integral commutativity is sufficient to write the solution of \eqref{eq35} in the closed form according to \eqref{eq36}. However, these conditions are not necessary. It may happen that a generator of evolution satisfies neither of the two conditions, but one is still able to write the solution of the dynamics equation in the closed form. To be more specific, in this article we shall investigate the Fedorov theorem, which demonstrates that a closed-form solution can be obtained under the condition of partial commutativity \cite{Fedorov1960} (for English refer to pp. 39--44 in \cite{Erugin1966}).

First, one should be reminded that every time-dependent linear generator $\mathbb{L}(t)$ can always be represented as a matrix, which makes it possible to study the algebraic properties of the generator. On the other hand, the evolution equation given by \eqref{eq35} can always be transformed into a differential equation with the generator $\mathbb{L}(t)$ in its matrix form multiplying the vectorized density matrix $\mathrm{vec} \{ \rho(t)\}$, which is simpler from the computational point of view. The operator $\mathrm{vec} \{ \rho(t)\}$ should be understood as a vector constructed by stacking the columns of $\rho(t)$ one underneath the other and such operation shall be referred to as the "$\mathrm{vec}$ operator". Thus, let us consider the master equation in the vectorized form, i.e.:
\begin{equation}\label{eq41a}
\mathrm{vec}\{ \dot{\rho}(t) \} = \mathbb{L}(t) \; \mathrm{vec}\{ \rho(t) \}
\end{equation}
and for such dynamics we shall formulate the Fedorov theorem \cite{Fedorov1960}.

\begin{theorem}[Fedorov theorem]
If the generator of evolution $\mathbb{L}(t)$ (its matrix representation) satisfies the condition:
\begin{equation}\label{eq41}
[\,\mathbb{L}(t), B^n (t)\,]\: \alpha = 0 \hspace{0.9cm} \forall \:n=1,2, 3, \dots,
\end{equation}
where $B(t) = \int_0^t \mathbb{L}(\tau) d \tau$ and $\alpha$ is a constant vector, then the solution of \eqref{eq41a} can be written in the closed form:
\begin{equation}\label{eq42}
\mathrm{vec} \{ \rho(t)\} = \exp[B(t)] \,\alpha
\end{equation}
\end{theorem}

\begin{proof}
There exists an obvious decomposition of $\exp[B(t)]$, i.e.:

\begin{equation}\label{eq43}
\exp[B(t)] = \sum_{m=0}^{\infty} \frac{1}{m!} \,B^m (t),
\end{equation}
which allows one to write a formula for the first derivative of $\exp[B(t)]$:
\begin{equation}\label{eq44}
\begin{aligned}
{}& \frac{d \exp[B(t)]}{d t}= \mathbb{L}(t) + \frac{1}{2!}\left\{\mathbb{L}(t) B(t) + B(t) \mathbb{L}(t)\right\} +\\
&+ \frac{1}{3!} \left\{ \mathbb{L}(t) B^2(t) + B(t) \mathbb{L}(t) B(t) + B^2(t) \mathbb{L}(t) \right \} + \dots= \\
& = \sum_{m=1}^{\infty} \frac{1}{m!} \sum_{k=1}^m B^{k-1} (t) \mathbb{L}(t) B^{m-k} (t).
\end{aligned}
\end{equation}
On the other hand, one can notice that the assumption (see \eqref{eq41}) can be transformed in the following way (for any $m,n \in \mathbb{N}$):
\begin{equation}\label{eq45}
\begin{aligned}
{}& \mathbb{L}(t) B^n(t) \,\alpha = B^n(t)\mathbb{L}(t) \,\alpha  \hspace{0.25cm}\Leftrightarrow\\& B^m (t)\mathbb{L}(t) B^n(t) \,\alpha = B^{n+m}(t)\mathbb{L}(t) \,\alpha
\end{aligned}
\end{equation}
Keeping in mind \eqref{eq44} and \eqref{eq45}, one can check whether $\mathrm{vec} \{ \rho(t) \} = \exp[B(t)] \alpha$ satisfies the evolution equation given by \eqref{eq41a}:
\begin{equation}\label{eq46}
\begin{aligned}
{}&\frac{\dd \: \mathrm{vec} \{ \rho(t) \}}{\dd t}  =\frac{\dd\: \exp[B(t)]\: \alpha}{\dd t} = \\
&  = \sum_{m=1}^{\infty} \frac{1}{m!} \sum_{k=1}^m B^{k-1} (t) \mathbb{L}(t) B^{m-k} (t) =\\
& = \mathbb{L}(t)  \sum_{m=1}^{\infty} \frac{1}{m!} \sum_{k=1}^m B^{m-1} (t)\: \alpha = \mathbb{L}(t)  \sum_{m=1}^{\infty} \frac{1}{m!} m B^{m-1} (t) \:\alpha =\\
& = \mathbb{L}(t) \sum_{m=1}^{\infty} \frac{1}{(m-1)!} B^{m-1} (t) \:\alpha = \mathbb{L}(t) \sum_{m=0}^{\infty} \frac{1}{m!} B^{m} (t) \:\alpha =\\
& = \mathbb{L}(t)\: \exp[B(t)]\: \alpha = \mathbb{L}(t) \: \mathrm{vec} \{ \rho(t)\}.
\end{aligned}
\end{equation}
It means that $\mathrm{vec} \{ \rho(t) \}$ defined by the formula \eqref{eq42} satisfies the dynamics given by \eqref{eq41a}, which completes the proof.
\end{proof}

There are three issues connected with the Fedorov theorem that one should be aware of.

Firstly, the Fedorov theorem enables us to write the solution of the evolution equation in the closed form. However, there is a significant limitation -- as the initial vectors one can use only the vector (or vectors) $\alpha$ which satisfy the condition of partial commutativity introduced by \eqref{eq41}. Naturally, if one has two linearly independent vectors $\alpha_1$ and $\alpha_2$ and both of them satisfy \eqref{eq41}, then the linear combination of them $c_1 \alpha_1 + c_2 \alpha_2$ also satisfies the condition form the Fedorov theorem. Therefore, all vectors $\alpha$ which satisfy \eqref{eq41} constitute a subspace in the vector space. The subspace which contains all vectors $\alpha$ shall be denoted by $\mathcal{M}(\mathbb{L}(t))$.

Secondly, from the physical point of view, it is important to be able to determine the trajectory of the state on the basis of the solution of the evolution equation. However, it may happen that when one determines $\alpha$ satisfying the condition \eqref{eq41} for a specific generator of evolution, it turns out that after de-vectorization $\alpha$ is not a proper density matrix. In such a case, the solution with $\alpha$ as the initial vector is not a legitimate state trajectory. For this reason, from the physical point of view, it is required to use as the initial vectors only such $\alpha$ which belongs to the intersection $\mathcal{M}(\mathbb{L}(t)) \cap \mathrm{vec}\{S(\mathcal{H})\}$, where $\mathrm{vec}\{S(\mathcal{H})\}$ refers to the state set of all vectorized density matrices associated with the Hilbert space $\mathcal{H}$.

Thirdly, in practice, there is no need to take into account in \eqref{eq41} all powers of $B^n(t)$ up to infinity because one can always use the Cayley-Hamilton theorem\cite{Hamilton1853,Cayley1858,Frobenius1878}, which states that every matrix satisfies its characteristic polynomial. Therefore, if $B(t)$ is a $\mu \times \mu$ matrix, the $\mu$-th power of $B(t)$ linearly depends on the lower powers. Thus, in general, it is sufficient to consider the powers of $B(t)$ up to $\mu-1$. The number of necessary powers may be additionally reduced provided one can determine the degree of the minimal polynomial of $B(t)$, which can be done numerically for some generators $\mathbb{L}(t)$.

In the context of the Fedorov theorem, it is important to explain how the vectors $\alpha$ satisfying the condition \eqref{eq41} can be obtained. One should notice that we are searching for the subspace which can be expressed by the following formula:
\begin{equation}\label{eq47}
\mathcal{M}(\mathbb{L}(t)) := \bigcap_{n=1}^{\mu-1} \mathrm{Ker} [\,\mathbb{L}(t), B^n(t)\,].
\end{equation}

The formula \eqref{eq47} cannot be easily calculated, however, one might notice a significant similarity between this issue and the problem of finding common eigenvectors of two matrices \cite{Shemesh1984,Jamiolkowski2014}. Therefore, in the context of the Fedorov theorem, one can use the approach introduced by Shemesh in order to transform the formula for the subspace $\mathcal{M} (\mathbb{L}(t))$ into an expression, which will be straightforward in computing. Let us first prove a lemma.

\begin{Lemma}\label{lemma}
For any set of linear operators $\{R_1, \dots, R_{\kappa}\}$ a following relation holds true:
\begin{equation}\label{lemma1}
\bigcap_{i=1}^{\kappa} \mathrm{Ker}\; R_i = \mathrm{Ker} \sum_{i=1}^{\kappa} R_i ^{\dagger}\, R_i,
\end{equation}
where $R_i ^{\dagger}$ denotes the operator dual to $R_i$.
\end{Lemma}

\begin{proof}
Let us prove the lemma for two operators $R_1$ and $R_2$ since one can easily generalize the reasoning for a higher number of operators. Then, on the left-hand side of \eqref{lemma1}, we have $\mathrm{Ker} R_1  \cap \mathrm{Ker} R_2$. Next, we observe:
\begin{equation}
\begin{aligned}
x \in \mathrm{Ker} R_1  \cap \mathrm{Ker} R_2 \:\Leftrightarrow\: {}& x \in \mathrm{Ker} R_1 \land  x \in \mathrm{Ker} R_2\\
&R_1 \, x = 0 \land R_2 \, x = 0 \\
&R_1^{\dagger} R_1 \,x = 0 \land R_2^{\dagger} R_2 \,x = 0 \\
&\left(R_1^{\dagger} R_1 + R_2^{\dagger} R_2 \right)\, x = 0 \\
& x \in \mathrm{Ker} \left(R_1^{\dagger} R_1 + R_2^{\dagger} R_2 \right)
\end{aligned}
\end{equation}
and the last part finishes the proof.
\end{proof}

Based on the lemma \ref{lemma}, we can conclude that the closed-form solution according to \eqref{eq42} can be obtained for the initial vectors $\alpha$ which belong to the subspace $\mathcal{M} (\mathbb{L}(t))$ such that:
\begin{equation}\label{eq49}
\mathcal{M}(\mathbb{L}(t)) = \mathrm{Ker} \sum_{n=1}^{\mu-1}\, [\,\mathbb{L}(t), B^n(t)\,]^{\dagger} [\,\mathbb{L}(t), B^n(t)\,].
\end{equation}

To sum up, if one wants to apply the Fedorov theorem in order to obtain a closed-form solution of a differential equation with a time-dependent generator $\mathbb{L}(t)$, one needs to prove that the subspace $\mathcal{M}(\mathbb{L}(t))$ defined by \eqref{eq49} is non-empty, which can be done effectively thanks to the Shemesh criterion. Then, one can write a closed-form solution of the evolution equation: $\mathrm{vec} \{ \rho(t)\} = \exp[B(t)] \alpha$. This solution generates a legitimate trajectory from the physical point of view only if the initial vector $\alpha$ can be considered a vectorized density matrix, i.e. $\alpha \in \mathcal{M} (\mathbb{L}(t)) \cap \mathrm{vec}\{S(\mathcal{H})\}$. Generators $\mathbb{L}(t)$ such that the corresponding subspace $\mathcal{M}(\mathbb{L}(t))$ is non-empty can be called \textit{partially commutative}.

\section{Fedorov theorem in dynamics of open quantum systems}\label{section3}

\subsection{Preliminaries}

In this article, we shall consider the evolution generator $\mathbb{L}(t)$ of $d-$level quantum systems in the form \cite{Breuer2004,Grigoriu2013}:
\begin{equation}\label{e3.1}
\mathbb{L}(t)\, [\rho] = - i \left[H, \rho \right] + \sum_k \gamma_k (t) \left(V_k \rho V_k^{\dagger} - \frac{1}{2} \{V_k ^{\dagger} V_k, \rho \}  \right),
\end{equation}
which can be regarded as a specific type of time-dependent GKSL generator \cite{Gorini1976,Lindblad1976}, such that the jump operators $V_k$ are represented by constant matrices while the relaxation rates $\gamma_k (t)$ are time-dependent. The operator $H$ is hermitian, i.e. $H^{\dagger} = H$, and can be interpreted as the effective Hamiltonian which accounts for the unitary evolution. This generator preserves the Hermiticity and trace of the density matrix, but for negative relaxation rates in some time intervals the evolution features non-Markovian effects \cite{Breuer2009}. For this reason, we shall restrict our analysis only to the relaxation rates such that $\gamma_k (t) \geq 0$ for all $t\geq0$ and for any $k$, which means that the evolution may be called time-dependent Markovian (though the corresponding dynamical map is not a semigroup).

One of the algebraic methods used in the analysis is the technique to obtain a matrix representation of the generator of evolution. Such a procedure is feasible if we apply the property connected with the $\mathrm{vec}$ operator. For any three matrices $A, B, C$ such that their product $ABC$ is computable we have the following relation \cite{Roth1934}:
\begin{equation}\label{e3.2}
\mathrm{vec} \:( ABC ) = (C^T \otimes A) \:\mathrm{vec} B,
\end{equation}
which shall be called the Roth's column lemma. This property has been excessively studied within the field of pure mathematics \cite{Neudecker1969,Hartwig1975,Henderson1981} as well as applied to Physics in order to search for matrix representations of given GKSL generators of evolution \cite{Egger2014,Czerwinski2016,Czerwinski2016a}. Taking into account the Roth's column lemma \eqref{e3.2}, one transforms the generator of evolution given originally by \eqref{e3.1} into the matrix form
\begin{equation}\label{e3.3}
\begin{aligned}
&\mathbb{L}(t) = i \left( H^T  \otimes \mathbb{1}_d - \mathbb{1}_d \otimes H  \right) + \\& + \sum_{k} \gamma_k  (t)  \left ( \overline{V}_k \otimes V_k - \frac{1}{2} \mathbb{1}_d \otimes V_k ^{\dagger} V_k - \frac{1}{2} V_k ^T \overline{V}_k \otimes \mathbb{1}_d \right ),
\end{aligned}
\end{equation}
where $\overline{V}_k$ denotes the complex conjugate of the jump operator $V_k$.

In our analysis, we consider three specific generators of evolution which govern the dynamics of three-level systems: $V-$system, \textit{cascade} and $\Lambda$-system \cite{Hioe1982}. For years such types of dynamics have been an important field of research since they are connected to optimal control of quantum dissipative systems in the context of laser cooling \cite{Rooijakkers1997,Tannor1999,Sklarz2004}. Therefore, we assume that $d = 3$ and the vectors $\{\ket{1}, \ket{2}, \ket{3}\}$ stand for the standard basis in the Hilbert space $\mathcal{H}$. A jump operator $V_k$ which corresponds to the transition form $j-$th level to $i-$th level shall be defined as $V_k := \ket{i} \bra{j} \equiv E_{ij}$.

As far as four-level systems are concerned ($d=4$), the standard basis is denoted by $\{\ket{1}, \ket{2}, \ket{3}, \ket{4}\}$. We demonstrate that one can define \textit{cascade}-type of evolution with $3$ jump operators accompanied by time-dependent decoherence rates, and then apply the Fedorov theorem to search for the dynamical map.

\subsection{Three-level $V-$system}

Three-level $V-$system relates to a physical scenario when an atom has two excited levels denoted by $\ket{1}$ and $\ket{3}$, but one ground state $\ket{2}$. The dynamics describes a decay from one of the excited level into the ground state. Thus, we have two jump operators: $E_{21} := \ket{2}\bra{1}$ and $E_{23} := \ket{2}\bra{3}$. We assume that the corresponding decoherence rates are given by the functions: $\gamma_{21} (t) := \mathrm{sin}^2 \omega t$ and $\gamma_{23} (t) := \mathrm{cos}^2 \omega t$. Then, based on the Roth's column lemma, the matrix form of the generator can be found according to \eqref{e3.3}:
\begin{equation}\label{e3.4}
\begin{aligned}
\mathbb{L}_V (t) {}&=  i \left( H_V^T  \otimes \mathbb{1}_3 - \mathbb{1}_3 \otimes H_V \right) + \\ &+ \mathrm{sin}^2\omega t \left (E_{21} \otimes E_{21} - \frac{1}{2} \mathbb{1}_3 \otimes E_{11} -\frac{1}{2} E_{11} \otimes \mathbb{1}_3  \right) + \\
& + \mathrm{cos}^2 \omega t \left(E_{23} \otimes E_{23} - \frac{1}{2} \mathbb{1}_3 \otimes E_{33} - \frac{1}{2} E_{33} \otimes \mathbb{1}_3  \right),
\end{aligned}
\end{equation}
where $H_V$ denotes the unperturbed Hamiltonian which describes three energy levels of the $V-$system, i.e. $H_V = \mathrm{diag} ( \mathcal{E}_1, 0, \mathcal{E}_3 )$ (the energy of the ground level is normalized to zero, i.e. $\mathcal{E}_2= 0$).

One can check that the generator for the $V-$system satisfies the following relations:
\begin{equation}\label{e3.5}
\begin{split}
[\mathbb{L}_V (t), \: \mathbb{L}_V (\tau) ] = 0 \hspace{0.5cm}\forall \: t,\tau \geq 0\\
[ \mathbb{L}_V (t) ,  \int \mathbb{L}_{V} (t) \, d t] = 0,
\end{split}
\end{equation}
which implies that the closed-form solution of the evolution equation can be obtained based on the Lappo-Danilevsky criterion (without the Fedorov generalization):
\begin{equation}\label{vdynamics}
\rho (t) = \exp \left(\int_0^t \mathbb{L}_V (\tau) d \tau \right) [\rho(0)].
\end{equation}

Let us investigate, as a specific example, the trajectory of the initial state: $\rho(0) = 1/2 \ket{1} \bra{1} + 1/2 \ket{3} \bra{3}$, which corresponds to a statistical mixture of two excited states with equal probabilities. The trajectory of this state can be described by a following dynamical map:
\begin{equation}\label{e3.6}
\begin{split}
&\rho(t) =\\ & \begin{pmatrix} \frac{1}{2} e^{\frac{-2 \omega t +  \mathrm{sin}(2 \omega t)}{4 \omega}} & 0 & 0 \\
0 & 1 - e^{-\frac{t}{2}} \:\mathrm{Cosh}\left[ \frac{\mathrm{sin} (2 \omega t)}{4 \omega} \right]  & 0 \\ 0 & 0 & \frac{1}{2} e^{- \frac{2 \omega t +  \mathrm{sin}(2 \omega t)}{4 \omega}} \end{pmatrix}.
\end{split}
\end{equation}
In order to study in detail the dynamics governed by the generator \eqref{e3.4}, let us consider the probability of finding the quantum system in each of the possible states as a function of time. By $p_i (t) $ we denote the probability of finding the system in $i-$th state at time instant $t$. One can find the plots in \figureref{vplot}. 

\begin{figure}[h!]
	\centering
	\begin{tabular}{c}
		\centered{\includegraphics[width=0.95\columnwidth]{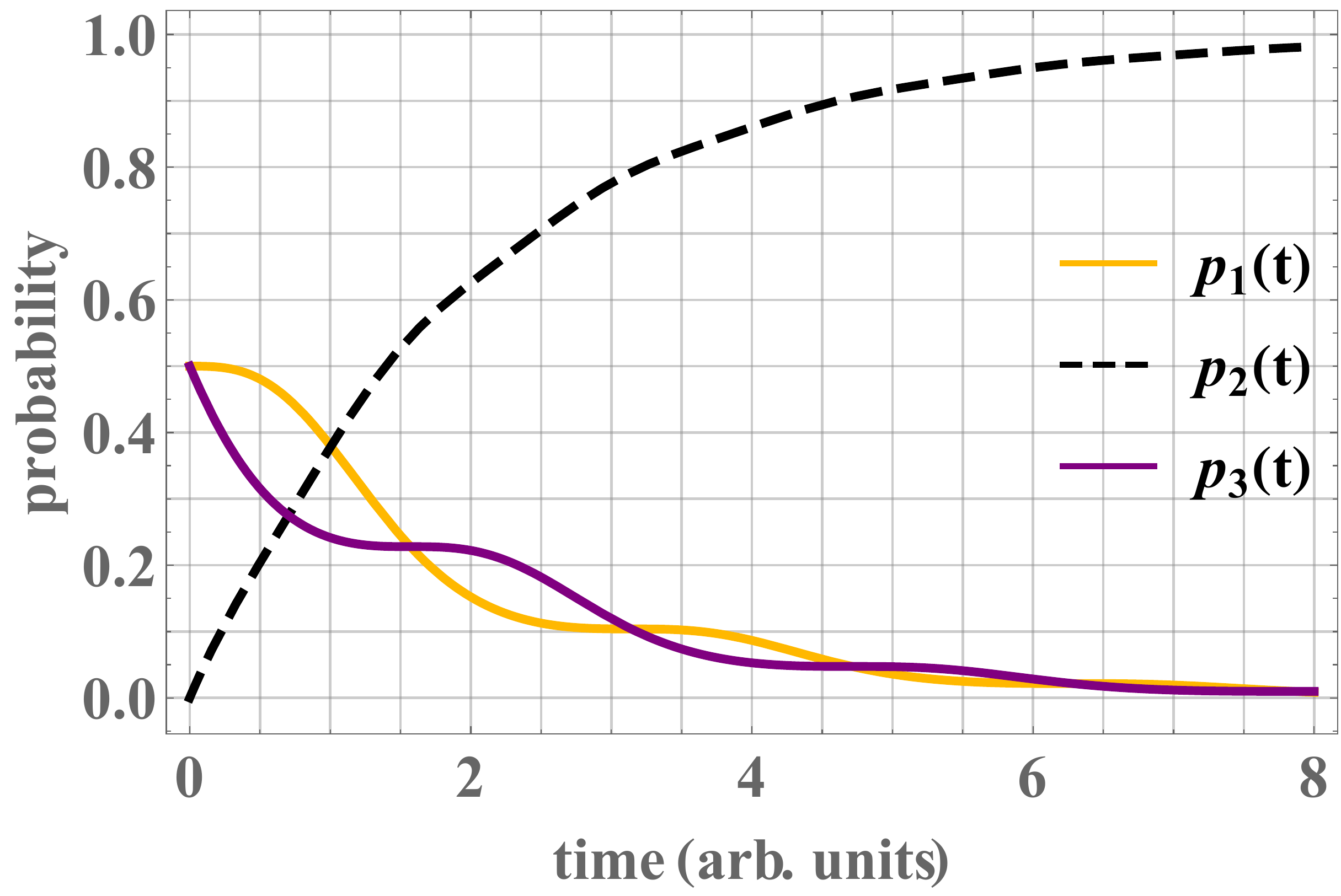}} \\
	\end{tabular}
	\caption{Plots present the probability of finding the three-level $V-$system in one of the possible states.}
	\label{vplot}
\end{figure}

One can observe that the probability of finding the system in the state $\ket{2}$ is an increasing function the value of which asymptotically converges to $1$. It is not an unexpected result since the $V-$model describes a three-level system which decays into the ground state in time. Nonetheless, it is worth noting that the probabilities $p_{1} (t) $ and $p_{3} (t)$ present specific shapes due to the fact that we introduced the oscillating functions (i.e. $\mathrm{sin}\, \omega t$ and $\mathrm{cos} \,\omega t$) into the decoherence rates. One could exchange the relaxation rates of the generator \eqref{e3.4} into different time-dependent functions and then explore other time characteristics of the probabilities.

In order to investigate more effects, one can add phase factors into the off-diagonal elements of the initial density matrix, i.e. $\rho_{13} (0) = 1/2\, e^{- i \phi}$ and $\rho_{31} (0) = 1/2\, e^{ i \phi}$, where $\phi$ stands for the relative phase between the states $\ket{1}$ and $\ket{3}$. Such a generalization does not affect the formulas for probabilities as presented in \figureref{vplot}, but allows one to additionally study how the phase factors change in time. Then, by applying the dynamics \eqref{vdynamics}, one would obtain:
\begin{equation}
\rho_{13} (t) = \frac{1}{2}\, e^{\left(-1/2+i (\mathcal{E}_3-\mathcal{E}_1) \right) t} \, e^{- i \phi} \hspace{0.25cm}\text{and}\hspace{0.25cm} \rho_{31} (t) =  \overline{\rho_{13} (t)},
\end{equation}
which means that the relative phase $\phi$ between the energy states $\ket{1}$ and $\ket{3}$ vanishes while the initial state decays into the ground level $\ket{2}$. The phase-damping effect is caused by the factor $e^{-1/2 t}$, whereas the other coefficient emerging from the evolution, i.e. $e^{ i  (\mathcal{E}_3-\mathcal{E}_1) t}$, makes the phase factor rotate on the complex plane. For arbitrary $\mathcal{E}_3$ and $\mathcal{E}_1$, the time evolution of the phase factor $\rho_{31} (t)$ is presented in \figureref{phase}.

\begin{figure}[h!]
	\centering
	\begin{tabular}{c}
		\centered{\includegraphics[width=0.95\columnwidth]{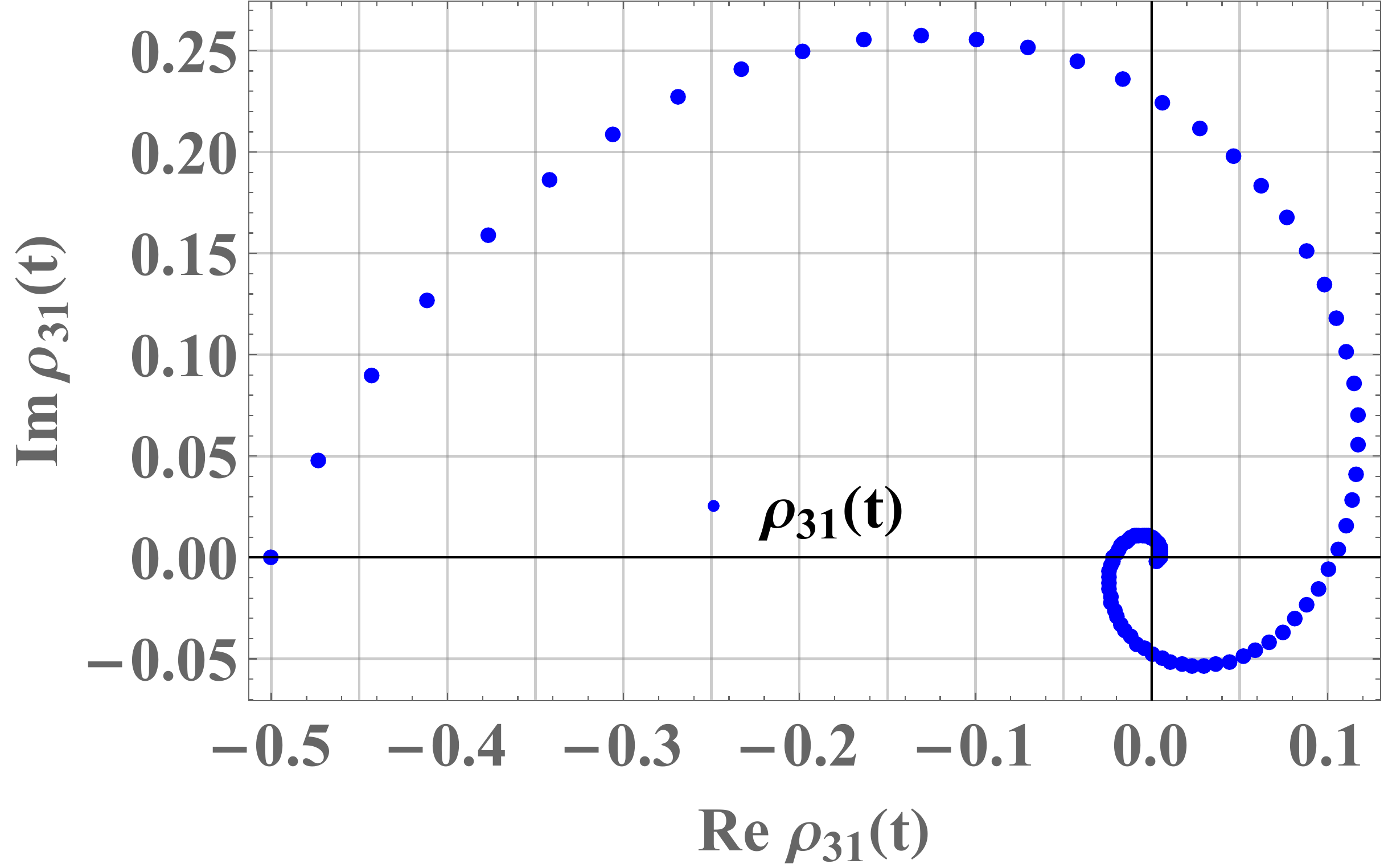}} \\
	\end{tabular}
	\caption{Plot presents the trajectory of $\rho_{31} (t)$ on the complex plane, assuming that the initial value of the relative phase equals $\pi$.}
	\label{phase}
\end{figure}

\subsection{Three-level \textit{cascade} system}

The tree-level model called \textit{cascade} describes a situation when the system can relax from the state $\ket{3}$ into the middle level $\ket{2}$ and then into the ground state denoted by $\ket{1}$. Since two kinds of transition are admissible, we have two jump operators: $E_{23} := \ket{2} \bra{3}$ and $E_{12} := \ket{1}\bra{2}$. We assume that the corresponding relaxation rates are again given by the functions: $\gamma_{23} (t) := \mathrm{sin}^2 \omega t$ and $\gamma_{12} (t) := \mathrm{cos}^2 \omega t$. This leads to the generator of evolution in the following representation:
\begin{equation}\label{e3.7}
\begin{aligned}
\mathbb{L}_C (t) {}&=  i \left( H_C^T  \otimes \mathbb{1}_3 - \mathbb{1}_3 \otimes H_C \right) + \\ & \mathrm{sin}^2\omega t \left (E_{23} \otimes E_{23} - \frac{1}{2} \mathbb{1}_3 \otimes E_{33} -\frac{1}{2} E_{33} \otimes \mathbb{1}_3  \right) + \\
& + \mathrm{cos}^2 \omega t \left(E_{12} \otimes E_{12} - \frac{1}{2} \mathbb{1}_3 \otimes E_{22} - \frac{1}{2} E_{22} \otimes \mathbb{1}_3  \right),
\end{aligned}
\end{equation}
where $H_C$ denotes the unperturbed Hamiltonian that describes three symmetric energy levels, i.e. $H_C = \mathrm{diag} (-\mathcal{E}, 0, \mathcal{E} )$ (the energy of the intermediate level is normalized to zero).

One can check that for the generator $\mathbb{L}_C (t)$ we obtain:
\begin{equation}\label{e3.8}
\begin{split}
[\,\mathbb{L}_C (t), \: \mathbb{L}_C (\tau)\,] \neq  0,\\
[\, \mathbb{L}_C (t), \: \int \mathbb{L}_C (t) \,d t\,] \neq 0,
\end{split}
\end{equation}
which implies that the sufficient conditions for the closed-form solution are not satisfied. Therefore, there is a need for a more general approach. One can consider the Fedorov theorem as a possible technique to solve the evolution equation with the generator \eqref{e3.7}.

In order to effectively apply the Fedorov theorem, we first need to numerically determine the minimal polynomial of $\int \mathbb{L}_C (t) \,d t$. The specific coefficients of the polynomial are of little interest since we focus on its degree which equals $6$. This means that for any $t\geq 0$ the operator $(\int \mathbb{L}_C (t) \,d t)^6$ can be expressed by means of the lower powers of $\int \mathbb{L}_{C} (t) \,\dd t$. Combining this observation with the earlier result \eqref{eq49}, we need to investigate the kernel of the operator:
\begin{equation}\label{e3.9}
\begin{aligned}
&\Gamma^{(C)} \equiv \\&
\sum_{n=1}^{5} \left[\mathbb{L}_C (t), \left(\int \mathbb{L}_{C} (t) \,\dd t\right)^n\right]^{\dagger} \left[\mathbb{L}_C (t), \left(\int \mathbb{L}_{C} (t) \,\dd t\right)^n\right].
\end{aligned}
\end{equation}
The matrix representation of $\Gamma^{(C)}$ can be found numerically. One can obtain that $\Gamma^{(C)}_{99} = g(t) \neq 0$ and all the other elements are equal zero. This means the intersection of $\mathrm{vec} \,\mathcal{S} (\mathcal{H})$ and $\mathcal{M} (\mathbb{L}_C (t)) = \mathrm{Ker} \,\Gamma^{(C)}$ can be written as:
\begin{equation}\label{e3.10}
\mathrm{vec}\, \rho \in \mathrm{vec} \,\mathcal{S} (\mathcal{H}) \cap \mathcal{M} (\mathbb{L}_C(t)) \:\Leftrightarrow\: \rho \in\mathcal{S} (\mathcal{H})  \land  \rho_{33} = 0,
\end{equation}
which implies that the evolution equation with the generator \eqref{e3.7} has a closed-form solution only for the initial states which assume zero probability for the level $\ket{3}$. Thus, the dynamical map can be written as:
\begin{equation}\label{e3.11}
\rho (t) = \exp \left(\int_0^t \mathbb{L}_C (\tau) d \tau \right) [\rho(0)],
\end{equation}
where $\rho (0) = p \, \ket{1} \bra{1} + (1-p) \,\ket{2} \bra{2}$ and $0\leq p \leq1$ (one may add phase factors on the off-diagonal elements). The explicit form of $\rho (t)$ can be computed:
\begin{equation}\label{e3.12}
\rho (t) = \begin{pmatrix} 1 - \xi (t) & 0 & 0\\ 0 & \xi (t)  & 0 \\ 0 & 0 & 0  \end{pmatrix},
\end{equation}
where
\begin{equation}\label{e3.13}
\xi (t) := \left(1 -p \right) \exp \left(- \frac{2 \omega t + \mathrm{sin} (2 \omega t)}{4 \omega}\right).
\end{equation}

In order to illustrate the results of the method, let us assume that $p=0$, i.e. the initial density matrix $\rho(0) = \ket{2} \bra{2}$. The plots in \figureref{cplot} present the probabilities $p_{1} (t) $ and $p_{2} (t)$ (naturally $p_{3} (t) = 0$ for all $t \geq 0$).

\begin{figure}[h!]
	\centering
	\begin{tabular}{c}
		\centered{\includegraphics[width=0.95\columnwidth]{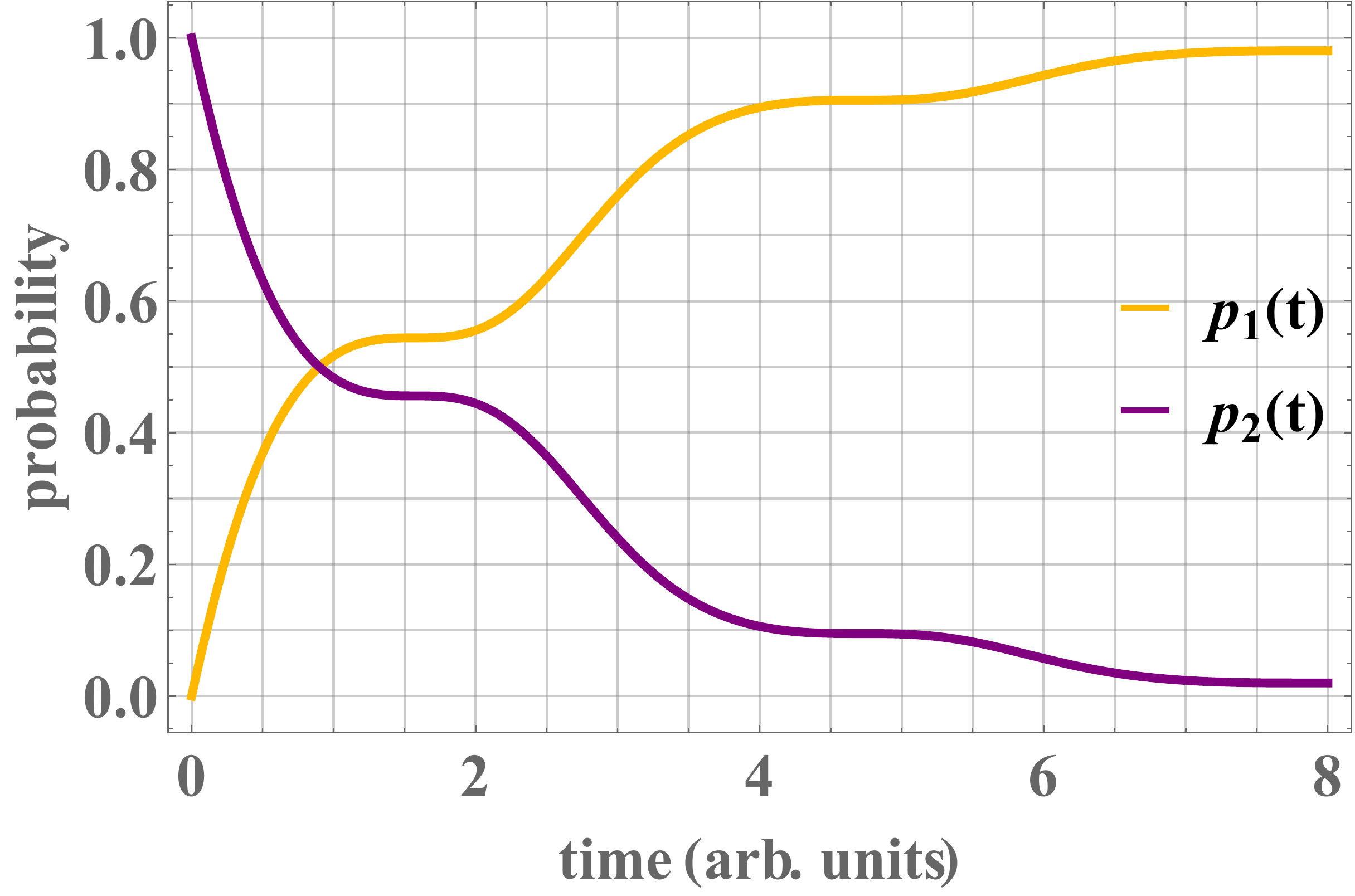}} \\
	\end{tabular}
	\caption{Plots present the probability of finding the three-level \textit{cascade} system in one of the possible states: $\ket{1}$ or $\ket{2}$.}
	\label{cplot}
\end{figure}

The results demonstrate the decay from the middle state $\ket{2}$ into the ground state $\ket{1}$ in time domain. The character of the probability graphs could by changed by modifying the functions which define the time-dependent relaxation rates: $\gamma_{23} (t)$ and $\gamma_{12} (t)$.

The process of relaxation within the \textit{cascade} model can also be analyzed by means of time-evolution of the purity and the von Neumann entropy. For a system described by a density matrix $\rho(t)$, the purity, which shall be denoted by $\pi(t)$, is defined as $\pi (t) := \tr \{\rho^2 (t)\}$. The von Neumann entropy has the standard definition: $S(t) := \tr \{\rho (t) \ln \rho^2 (t) \}$. Note that usually these figures are computed for a given state, whereas we treat them as the functions of time since we wish to follow the dynamics of entropy and purity for the initial state $\rho(0) = p \, \ket{1} \bra{1} + (1-p) \,\ket{2} \bra{2}$. We obtain the formulas:
\begin{equation}\label{e3.13a}
\begin{split}
&\pi (t) =  2 \xi^2 (t) - 2 \xi (t) +1,\\
&S(t) = - (1 - \xi (t)) \ln \{ 1 - \xi (t)\} -  \xi (t) \ln \{\xi (t)\}.
\end{split}
\end{equation}

To be more specific, let us again assume that $p=0$. And for the initial state $\rho(0) = \ket{2} \bra{2}$ we can plot the functions: $\pi (t)$ and $S(t)$ \figureref{purity}.

\begin{figure}[h!]
	\centering
	\begin{tabular}{c}
		\centered{\includegraphics[width=0.95\columnwidth]{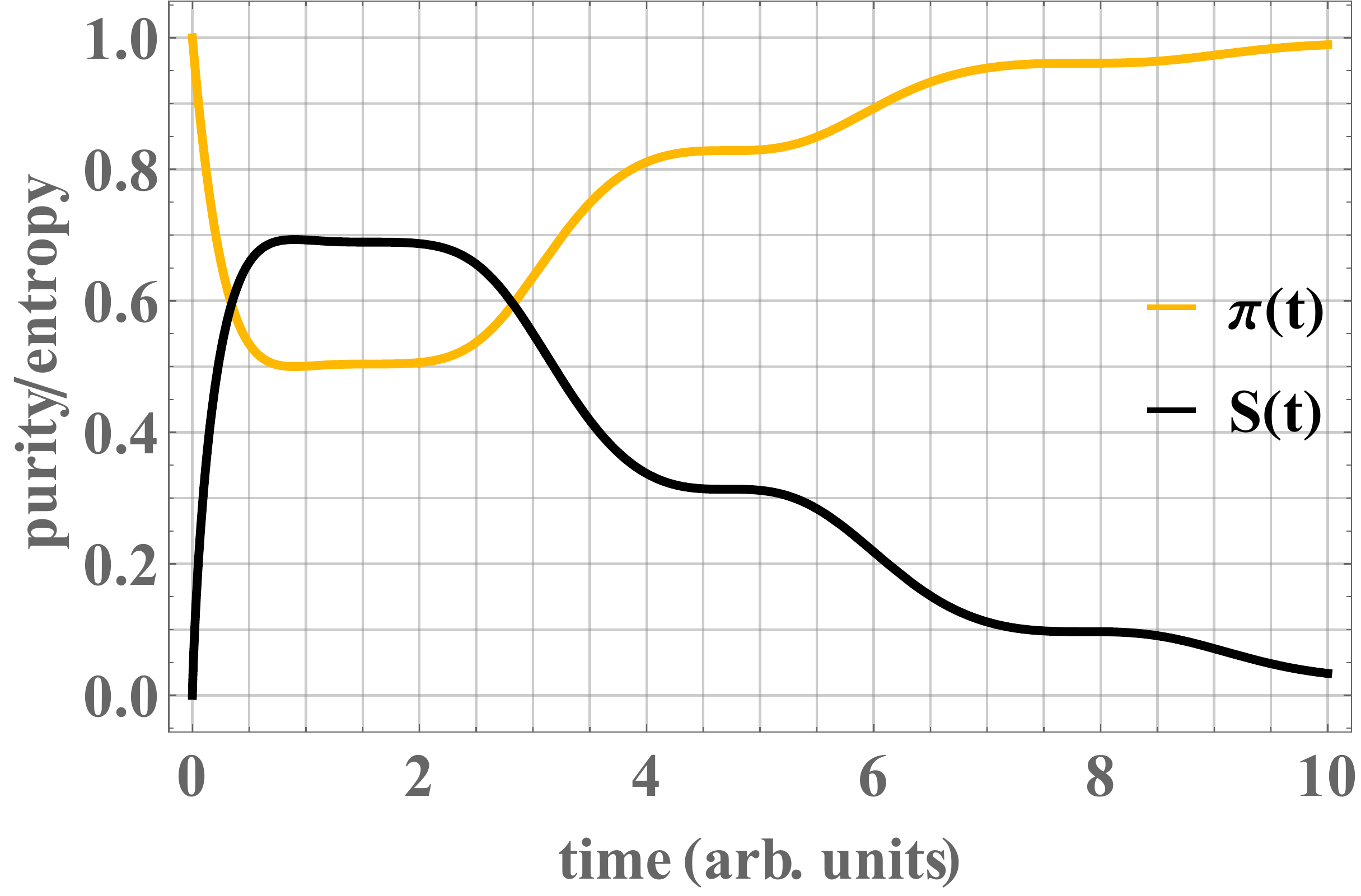}} \\
	\end{tabular}
	\caption{Plots present the purity $\pi (t)$ and the von Neumann entropy $S(t)$ of a dissipative system subject to \textit{cascade} decoherence model.}
	\label{purity}
\end{figure}

Since the input was a pure state, we have $\pi(0)= 1$ and $S(0) = 0$. Then, the state is getting more mixed with time. At some point, we have equal probabilities for $\ket{2}$ and $\ket{1}$, which means that the purity drops down to its minimal value, i.e. $\pi (t') = 1/2$ whereas the von Neumann entropy reaches its maximum value $S(t') = \ln 2 \approx 0.69315$. In time, both functions are approaching to their initial values since the final state is also pure. The shape of the functions reflects the definitions of the relaxation rates.

One can also consider time-evolution of off-diagonal elements of the density matrix by imposing a relative phase $\phi$ between the states $\ket{1}$ and $\ket{2}$. Then, the initial density matrix $\sigma (0)$ can be introduced in the form:
\begin{equation}\label{input2}
\sigma (0)  =  \frac{1}{2} \begin{pmatrix} 1 & e^{- i \phi} & 0 \\ e^{ i \phi} &1  & 0 \\  0 & 0 & 0  \end{pmatrix}.
\end{equation}

Such a change in the initial density matrix allows one to study dynamics of the phase factors. Based on the dynamical map \eqref{e3.11}, we obtain:
\begin{equation}
\sigma_{12} (t) = \frac{1}{2} \exp\left( \left(-\frac{1}{4} + \mathcal{E} i \right)t -\frac{\sin 2 \omega t}{8 \omega} \right) e^{- i \phi} \hspace{0.25cm}
\end{equation}
and $\sigma_{21} (t) = \overline{\sigma_{12} (t)}$, which gives the trajectory of the phase factor as presented in \figureref{phase2} (for arbitrary $\omega$ and $\mathcal{E}$).

\begin{figure}[h!]
	\centering
	\begin{tabular}{c}
		\centered{\includegraphics[width=0.95\columnwidth]{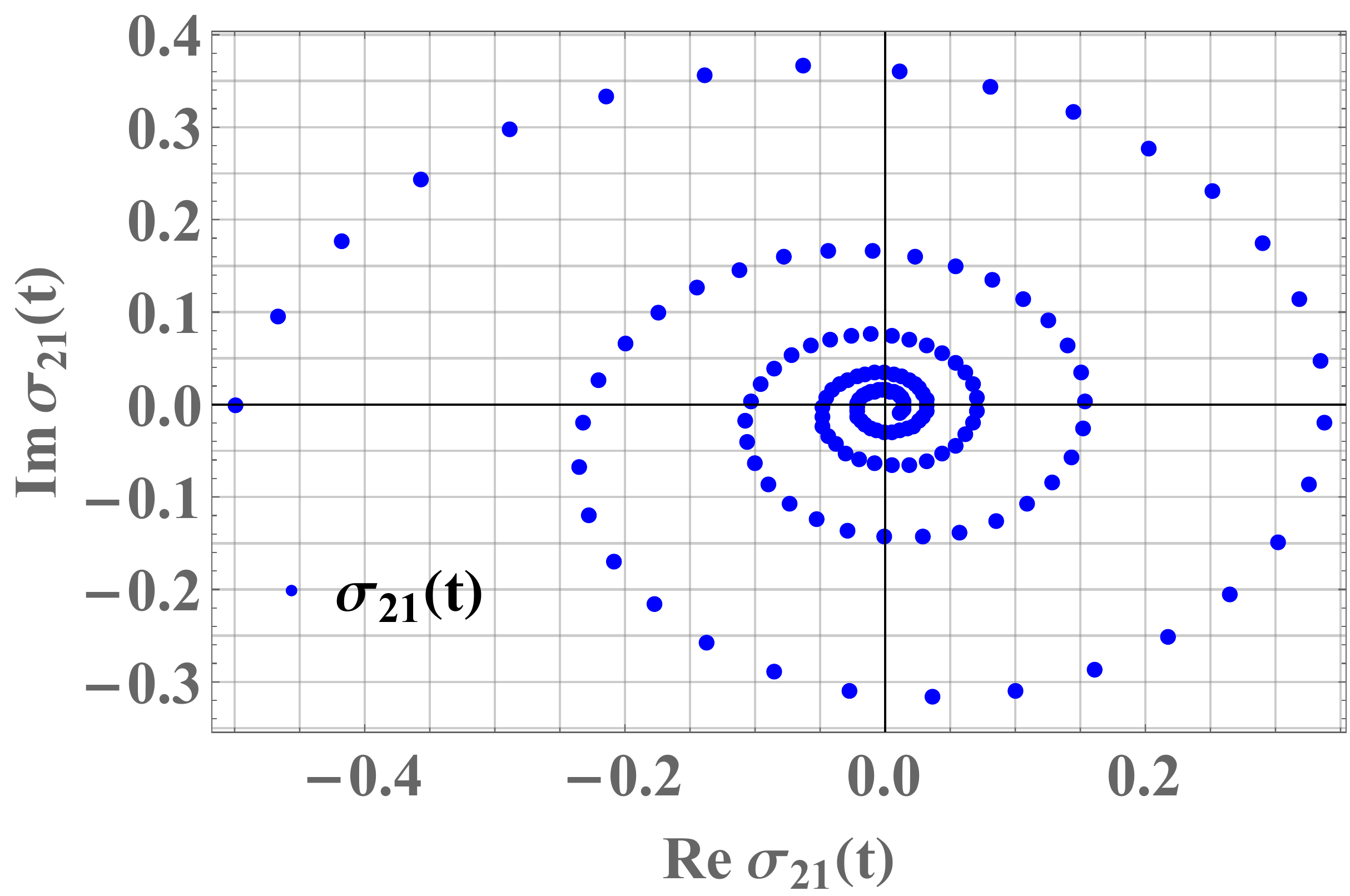}} \\
	\end{tabular}
	\caption{Plot presents the trajectory of $\sigma_{21} (t)$ on the complex plane, assuming that the initial value of the relative phase equals $\pi$.}
	\label{phase2}
\end{figure}

\subsection{Three-level $\Lambda$-system}

Quantum $\Lambda-$system with three energy levels belongs to very useful models studied in different areas of modern Physics, e.g. \cite{Brion2007,Zhou2016,Parshkov2018}. It is assumed that the system decays from the excited level $\ket{2}$ into one of two lower-energy states: $\ket{1}$ or $\ket{3}$. Thus, we have two jump operators: $E_{12} := \ket{1}\bra{2}$ and $E_{32} := \ket{3}\bra{2}$. We shall consider the following generator of evolution:
\begin{equation}\label{e3.14}
\begin{aligned}
\mathbb{L}_{\Lambda} (t) {}&= i \left( H_{\Lambda}^T  \otimes \mathbb{1}_3 - \mathbb{1}_3 \otimes H_{\Lambda} \right) + \\ &+  f_1 (t) \left (E_{12} \otimes E_{12} - \frac{1}{2} \mathbb{1}_3 \otimes E_{22} -\frac{1}{2} E_{22} \otimes \mathbb{1}_3  \right) + \\
& + f_2 (t)  \left(E_{32} \otimes E_{32} - \frac{1}{2} \mathbb{1}_3 \otimes E_{22} - \frac{1}{2} E_{22} \otimes \mathbb{1}_3  \right).
\end{aligned}
\end{equation}
where the functions $f_i (t): \mathcal{I} \rightarrow \mathbb{R}_+$ are assumed to be linearly independent and $H_{\Lambda}$ stands for the Hamiltonian which describes the energy levels, i.e.  $H_{\Lambda} = \mathrm{diag} (- \mathcal{E}_1, 0, -\mathcal{E}_3)$ for $\mathcal{E}_2, \mathcal{E}_3 > 0$. One can notice that this generator is not functionally commutative, neither it commutes with its integral. The minimal polynomial of \eqref{e3.14} cannot be easily determined without any assumptions concerning the functions: $f_1 (t), f_2 (t)$ and the energies: $\mathcal{E}_1, \mathcal{E}_3$, which means that in order to consider the Fedorov theorem in the context of $\Lambda-$systems we need to search for the kernel of:
\begin{equation}\label{e3.16}
\begin{aligned}
& \Gamma^{(\Lambda)} \equiv \\ & 
\sum_{n=1}^{8} \left[\mathbb{L}_{\Lambda}(t), \left(\int \mathbb{L}_{\Lambda}(t) \,d t\right)^n\right]^{\dagger} \left[\mathbb{L}_{\Lambda}(t), \left(\int \mathbb{L}_{\Lambda}(t) \,d t\right)^n\right].
\end{aligned}
\end{equation}
Interestingly, regardless of the functions: $f_1 (t), f_2 (t)$ and the energies: $\mathcal{E}_1, \mathcal{E}_3$, it can be checked numerically that $\Gamma^{(\Lambda)}_{55} \neq 0$ and all the other elements are zeros. For this reason, we can write
\begin{equation}\label{e3.17}
\mathrm{vec}\, \rho \in \mathrm{vec} \,\mathcal{S} (\mathcal{H}) \cap \mathcal{M} (\mathbb{L}_{\Lambda}(t) ) \:\Leftrightarrow\: \rho \in\mathcal{S} (\mathcal{H})  \land  \rho_{22} = 0,
\end{equation}
which means that the differential equation with the generator \eqref{e3.14} has a closed-form solution for example when the initial state is given by $\rho_S (0) = p \ket{1}\bra{1} + (1-p) \ket{3}\bra{3}$. However such a state, which is a statistical mixture of two lower-energy states, is stationary because the dynamics does not allow any transitions between the levels $\ket{1}$ and $\ket{3}$. Thus, for any functions $f_1 (t)$ and $f_2 (t)$, we have
\begin{equation}\label{e3.18}
\rho(t) = \exp \left(\int_0^t \mathbb{L}_{\Lambda} (\tau) d \tau \right) [\rho_S (0)] = \rho_S (0).
\end{equation}
Alternatively, one impose a relative phase between the states $\ket{1}$ and $\ket{3}$ and consider how the dynamics influence the off-diagonal elements. If we introduce the initial state in the form:
\begin{equation}\label{laminput}
\rho (0)  =  \frac{1}{2} \begin{pmatrix} 1 & 0 & e^{- i \phi} \\ 0 &0  & 0 \\ e^{ i \phi} & 0 & 1  \end{pmatrix},
\end{equation}
where $\phi$ stands for the relative phase, then one can observe that such initial state also satisfies the condition of partiall commutativity. If we impose the dynamical map $\exp \left(\int_0^t \mathbb{L}_{\Lambda} (\tau) d \tau \right)$ on the state \eqref{laminput}, we obtain:
\begin{equation}
\rho (t) =   \frac{1}{2} \begin{pmatrix} 1 & 0 & e^{(\mathcal{E}_3-\mathcal{E}_1)t \,i} e^{- i \phi} \\ 0 &0  & 0 \\ e^{(\mathcal{E}_1 - \mathcal{E}_3) t \, i} e^{ i \phi} & 0 & 1  \end{pmatrix},
\end{equation}
which means that the phase factor rotates on the complex plane in time. The oscillations of the phase factor are attributed solely to the unitary evolution. If the energy levels were degenerate, i.e. $\mathcal{E}_3 =\mathcal{E}_1=0$, then the input state \eqref{laminput} would be stationary.

\subsection{Four-level \textit{cascade} system}

The four-level \textit{cascade} model describes a physical situation when the system can relax from the highest state $\ket{4}$ into the lower level $\ket{3}$, then into the state $\ket{2}$, and finally into the ground state denoted by $\ket{1}$. Since three kinds of transition are admissible, we have $3$ jump operators: $E_{34} := \ket{3} \bra{4}$, $E_{23} := \ket{2} \bra{3}$ and $E_{12} := \ket{1}\bra{2}$. There are plenty of possible time-dependent decoherence rates that might be analyzed in the context of such dynamics. We shall assume that the corresponding relaxation rates are given by the functions: $\gamma_{34} (t) := e^{- \omega t} $ and $\gamma_{23} (t) = \gamma_{12} (t)= \mathrm{sin}^2 (3\, \omega t)$. This leads to the generator of evolution in the following representation:

\begin{equation}\label{e3.19}
\begin{aligned}
{}&\mathbb{L}_{FC} (t) = i \left( H_{FC}^T  \otimes \mathbb{1}_4 - \mathbb{1}_4 \otimes H_{FC} \right) + \\
&+ e^{- \omega t} \left (E_{34} \otimes E_{34} - \frac{1}{2} \mathbb{1}_4 \otimes E_{44} -\frac{1}{2} E_{44} \otimes \mathbb{1}_4  \right) + \\
& + \mathrm{sin}^2 (3\, \omega\, t) \left(E_{23} \otimes E_{23} - \frac{1}{2} \mathbb{1}_4 \otimes E_{33} - \frac{1}{2} E_{33} \otimes \mathbb{1}_4  \right)+\\
& + \mathrm{sin}^2 (3\, \omega\, t) \left(E_{12} \otimes E_{12} - \frac{1}{2} \mathbb{1}_4 \otimes E_{22} - \frac{1}{2} E_{22} \otimes \mathbb{1}_4  \right),
\end{aligned}
\end{equation}
where $H_{FC}$ denotes a four-level \textit{cascade} Hamiltonian. The energy levels are assumed to be symmetric, i.e. $H_{FC} = \mathrm{diag} (- \mathcal{E}_2, -\mathcal{E}_1, \mathcal{E}_1, \mathcal{E}_2 )$ for $\mathcal{E}_1, \mathcal{E}_2>0$.One can verify that the generator $\mathbb{L}_{FC} (t)$ satisfies neither the condition of functional commutativity nor commutativity with its integral. Therefore, it is desirable to search for other methods which can be used to solve the evolution equation governed by the generator \eqref{e3.19}.

We investigate the kernel of the operator $\Gamma^{(FC)}$ (cf. \eqref{e3.9}). The matrix representation of this operator can be determined numerically. One can then observe that $\Gamma^{(FC)}_{16\, 16} = g(t)$, whereas the other elements are zeros. This means the intersection of $\mathrm{vec} \,\mathcal{S} (\mathcal{H})$ and $\mathcal{M} (\mathbb{L}_{FC} (t)) = \mathrm{Ker} \,\Gamma^{(FC)}$ can be written as:
\begin{equation}\label{e3.20}
\mathrm{vec}\, \rho \in \mathrm{vec} \,\mathcal{S} (\mathcal{H}) \cap \mathcal{M} (\mathbb{L}_{FC} (t)) \:\Leftrightarrow\: \rho \in\mathcal{S} (\mathcal{H})  \land  \rho_{44} = 0,
\end{equation}
which implies that the evolution equation with the generator \eqref{e3.19} has a closed-form solution only for the initial states which assume zero probability for the level $\ket{4}$. In other words, we are able to follow the dynamics in closed form only if we reduce the dimension of the system by one. Then, the dynamical map can be written as:
\begin{equation}\label{e3.21}
\rho (t) = \exp \left(\int_0^t \mathbb{L}_{FC} (\tau) d \tau \right) [\rho(0)],
\end{equation}
where $\rho (0)$ denotes an initial state satisfying \eqref{e3.20}, e.g. $\rho (0)= q_1 \, \ket{1} \bra{1} + q_2 \,\ket{2} \bra{2} + q_3 \,\ket{3} \bra{3}$ and $\{q_1, q_2, q_3\}$ stands for a probability distribution (one may add phase factors on the off-diagonal elements).

Let us study a specific example of this kind of dynamics by assuming that the initial state has a form: $\rho (0) = 1/3 \,\ket{2} \bra{2} + 2/3 \,\ket{3} \bra{3}$. Based on the closed-form solution \eqref{e3.21} one can compute:
\begin{equation}
\begin{cases}
p_{1}(t) = 1 + \frac{1}{18 \omega} \left( e^{\frac{-6 \omega t + \mathrm{sin} (6\, \omega\, t)}{12 \omega}} \left( -6 (3+t) \omega + \mathrm{sin} (6\, \omega\, t)  \right) \right)\\
 \\
p_{2} (t) = \frac{1}{18 \omega} \left( e^{\frac{-6 \omega t + \mathrm{sin} (6\, \omega\, t)}{12 \omega}} \left( 6 (1+t) \omega - \mathrm{sin} (6\, \omega\, t)  \right) \right)\\
 \\
p_{3} (t) = \frac{2}{3}  e^{\frac{-6 \omega t + \mathrm{sin} (6\, \omega\, t)}{12 \omega}}
\end{cases}
\end{equation}
where $p_{k} (t)$, like before, stands for the probability of finding the system in $k-$th state. In order to track the changes that occur in the system during the evolution, the functions $p_{k} (t)$ are presented in \figureref{fourlevel}.

\begin{figure}[h!]
	\centering
	\begin{tabular}{c}
		\centered{\includegraphics[width=0.95\columnwidth]{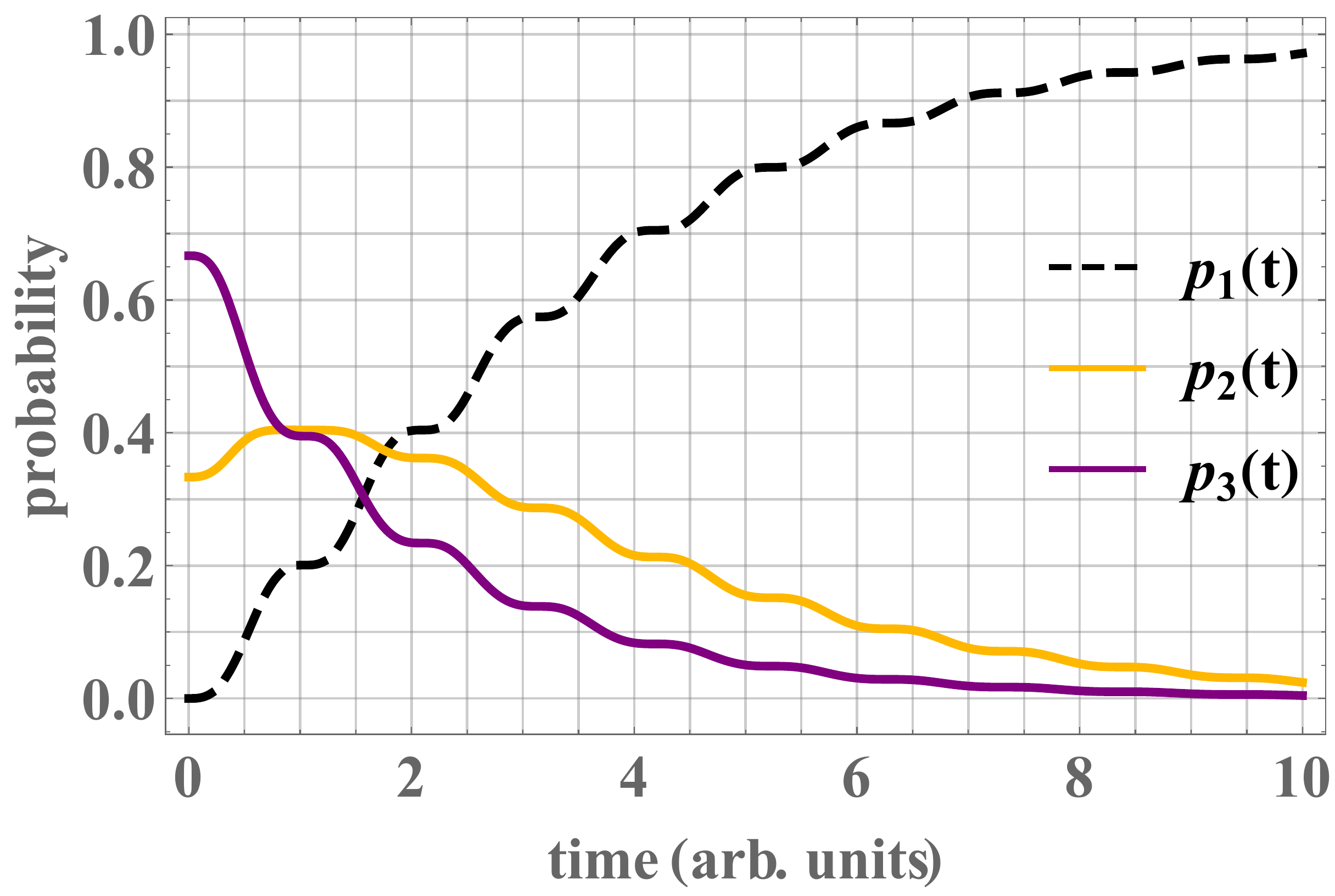}} \\
	\end{tabular}
	\caption{Plots present the probability of finding the four-level \textit{cascade} system in one of the possible states: $\ket{1}$, $\ket{2}$ or $\ket{3}$.}
	\label{fourlevel}
\end{figure}

Similarly as before, one can follow other characteristics of quantum system, such as the purity, denoted by $\pi (t)$, and the von Neumann entropy -- $S(t)$. In \figureref{figure5} one can observe the plots of these functions.

\begin{figure}[h!]
	\centering
	\begin{tabular}{c}
		\centered{\includegraphics[width=0.95\columnwidth]{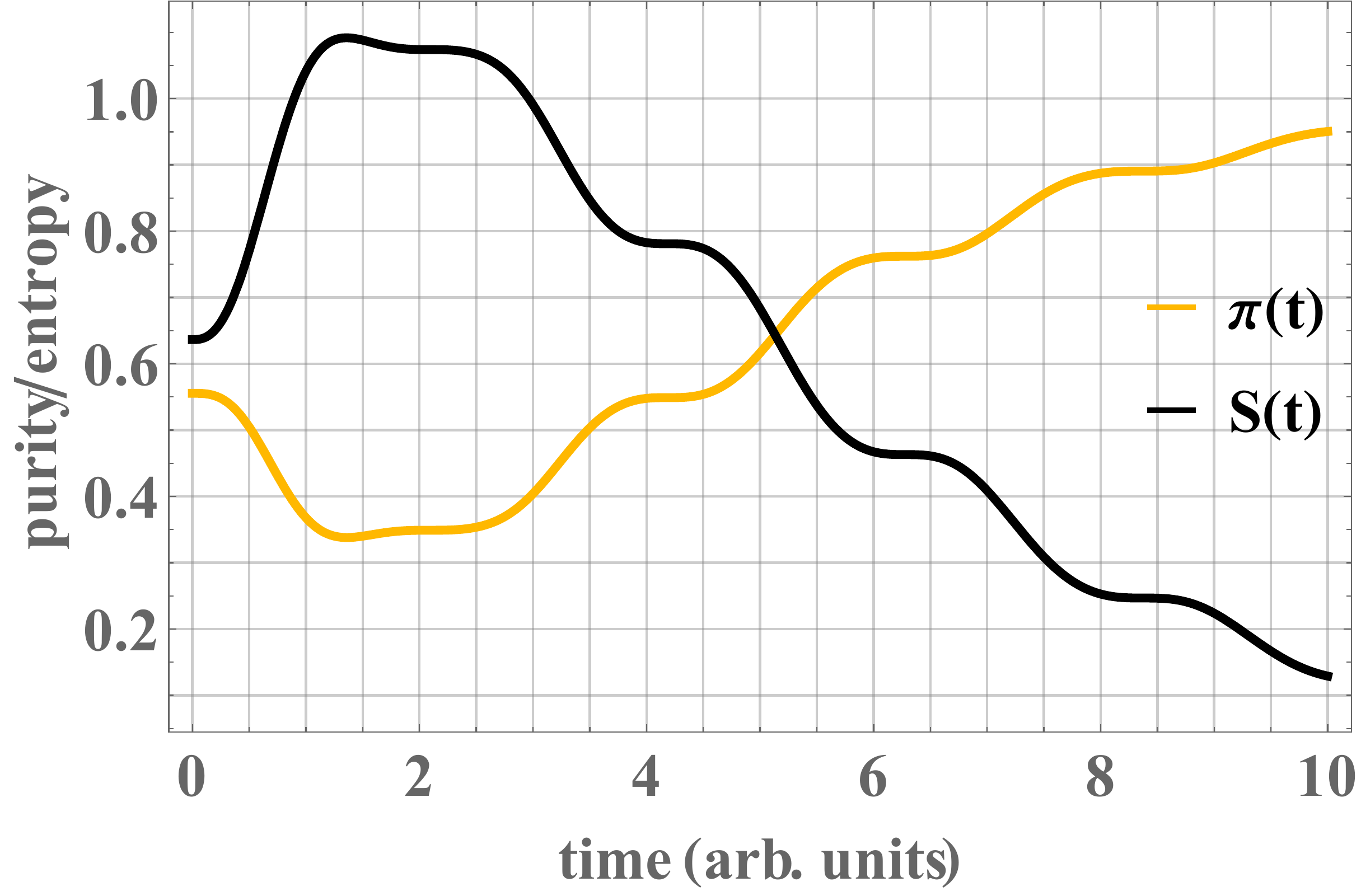}} \\
	\end{tabular}
	\caption{Plots present the purity and the von Neumann entropy of a dissipative four-level system subject to \textit{cascade} decoherence model.}
	\label{figure5}
\end{figure}

It is worth noting that one can choose any specific state satisfying the condition \eqref{e3.20} (e.g. with phase factors) and track its characteristics in time, assuming that the evolution is governed by the generator \eqref{e3.19}. For instance, we may consider a state in the form:
\begin{equation}\label{input3}
\sigma (0) = \frac{1}{3} \begin{pmatrix}  1 & e^{- i \,\phi_{12}} & e ^{- i \,\phi_{13}}&0 \\ e^{ i \,\phi_{12}} & 1 & e^{ i \,(\phi_{12}-\phi_{13})}&0 \\ e ^{ i\, \phi_{13}} & e^{ i \,(\phi_{13}-\phi_{12})} & 1&0 \\ 0&0&0&0 \end{pmatrix},
\end{equation}
where $\phi_{12}$ denotes the relative phase between the states $\ket{1}$ and $\ket{2}$ (and analogously for $\phi_{13}$). By applying the dynamical map \eqref{e3.21} to the state \eqref{input3}, we can determine the dynamics of the off-diagonal elements:
\begin{equation}
\begin{cases}
\sigma_{21} (t)  = \frac{1}{3} \exp \left( -\frac{1}{4} t + i (\mathcal{E}_1 - \mathcal{E}_2) t + \frac{\sin (6 \omega t) }{24 \omega} \right) e^{ i \phi_{12}} \\
 \\
\sigma_{31} (t) = \frac{1}{3} \exp \left(-\frac{1}{4} t - i (\mathcal{E}_1 + \mathcal{E}_2) t  + \frac{\sin (6 \omega t)}{24 \omega} \right) e^{ i \phi_{13}}\\
\\ 
\sigma_{32} (t) = \frac{1}{3} \exp \left( -\frac{1}{2} t - 2 \mathcal{E}_1 i t + \frac{\sin (6 \omega t)}{12 \omega}  \right) e^{i (\phi_{13}-\phi_{12})}
\end{cases}
\end{equation}
and from $\sigma_{ij} (t) = \overline{\sigma_{ji} (t)}$ we can get the other half. The trajectories can be presented graphically on the complex plane if we assume some arbitrary values of the parameters characterizing the evolution, i.e. $\omega,\mathcal{E}_1,\mathcal{E}_2$. For two exemplary phase factors it is done in \figureref{phase3}.

\begin{figure}[h!]
	\centering
	\begin{tabular}{c}
		\centered{\includegraphics[width=0.95\columnwidth]{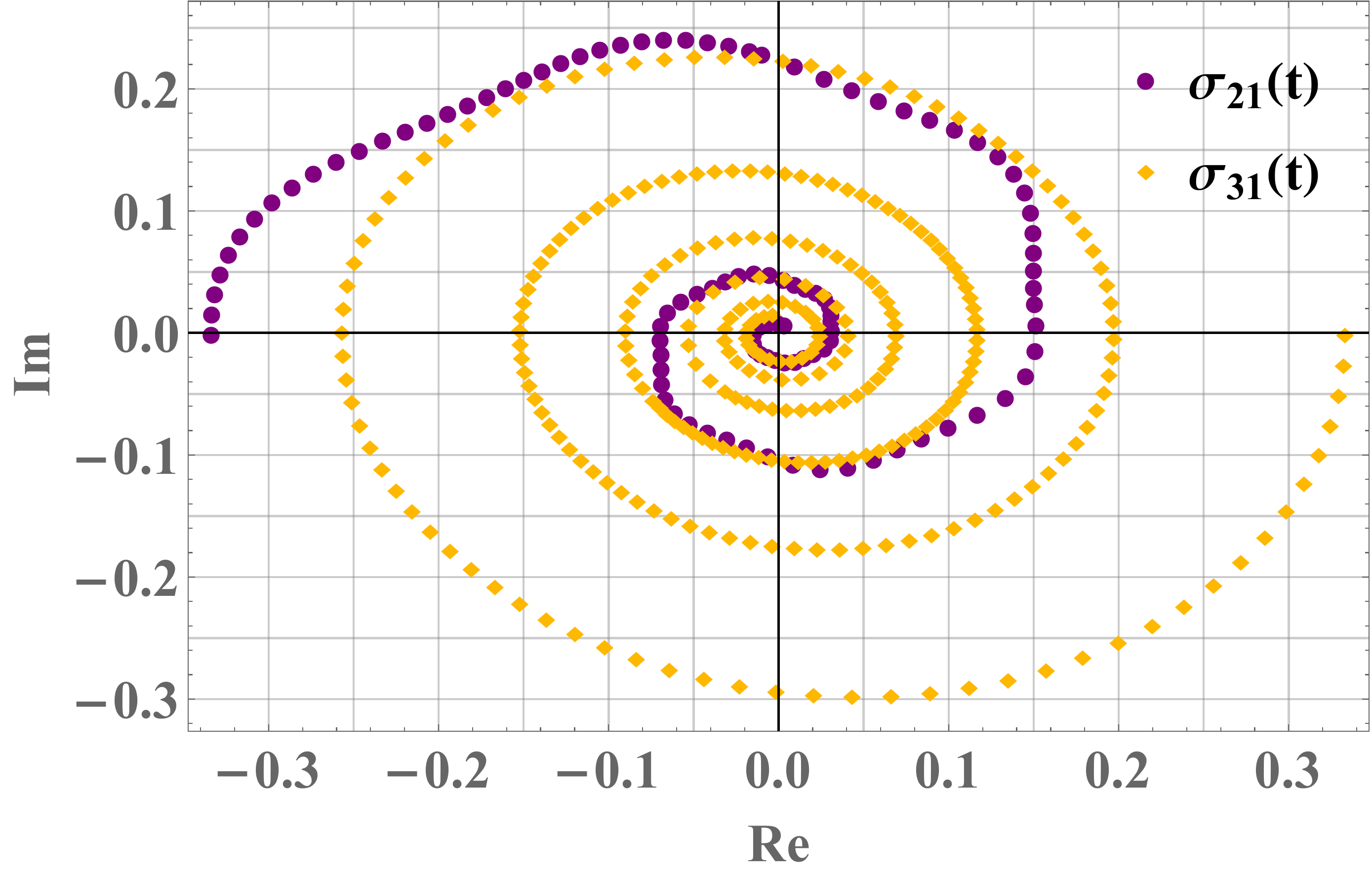}} \\
	\end{tabular}
	\caption{Plot presents the trajectories of $\sigma_{21} (t)$ and $\sigma_{31} (t)$ on the complex plane with the initial values of the relative phases: $ \phi_{12} = \pi$ and $\phi_{13} = 0$.}
\label{phase3}
\end{figure}

\subsection{Discussion and analysis}

The Fedorov theorem provides a useful generalization of the Lappo-Danilevsky criterion. This method was originally introduced by F. I. Fedorov in a 2-pages article in Russian \cite{Fedorov1960} and later included in the book by N. P. Erugin \cite{Erugin1966}. For a long time the theorem was unnoticed in the field of linear differential equations. However, in 2018 it was rediscovered by T. Kamizawa \cite{Kamizawa2018}, who proposed an effective analytical method for studying partial commutativity although with no reference to Physics.

This article contributes to the field of open quantum systems dynamics by demonstrating that the Fedorov theorem can be applied to search for dynamical maps if the corresponding generator depends on time. We considered three particular types of three-level dynamics: $V-$system, \textit{cascade} and \textit{Lambda} along with one example on four-level systems. Such evolution models are commonly studied in laser Physics.

In the case of the $V-$system, it turns out that the generator of evolution \eqref{e3.4} is functionally commutative (even if the relaxation rates are substituted with different time-dependent functions). This allows us to follow the trajectory for any initial state by the closed-form solution. For specific examples, we obtained plots which show how the probabilities of system being in basis states change in time. Interestingly, if one imposes a relative phase factor in the off-diagonal elements of the density matrix, we shall observe phase-damping effects which can be presented by trajectories of the phase factor on the complex plane.

The results for the \textit{cascade} model demonstrate that the Fedorov theorem can be useful but limited at the same time. The closed-form solution can be obtained only if there is zero probability for the initial state to be in the highest energy level. This means that we can study only the dynamics of a reduced, two-level subsystem. In spite of this limitation, one can determine the solution for a spectrum of density matrices and study time characteristics of the corresponding probabilities. The analysis can be further extended by analyzing the dynamics of the purity and the von Neumann entropy. In addition, one can analyze the dynamics of the off-diagonal elements of the density matrix by following the trajectories of phase factors on the complex plane.

Thirdly, in the case of the famous $\textit{Lambda}-$system, the Fedorov theorem allows one to write the solution only for such states which are stationary in terms of the probabilities. The system, given as a statistical mixture of the two lower states, remains unchanged subject to the generator of evolution. However, if we impose non-zero off-diagonal elements of the initial density matrix, we can observe oscillations of the phase factor, which is attributed to the unitary part of the generator.

Finally, an example of four-level systems with \textit{cascade} dynamics was studied. Based on the Fedorov theorem, we could obtain a closed-form solution for three-level subset of initial states. Dynamics of such states can be investigated by following the probabilities, purity, von Neumann entropy, as well as the trajectories of phase factors.

The examples studied in the article show that the applicability of the Fedorov theorem depends on the algebraic structure of the generator $\mathbb{L}(t)$. For some types of dynamics the Fedorov theorem may allow one to obtain a closed-form solution and track the time changes in quantum systems. This problem requires further research. More kinds of time-dependent generators should be tested in connection with the Fedorov theorem. Multi-level quantum systems subject to relaxation (e.g. laser cooling) are an area of intensive research, both theoretical and experimental, e.g. \cite{Bartana1993,Bartana1997,Yang2018}. The Fedorov theorem can provide an effective framework to study dynamics of such systems.

\section{Summary and outlook}

In the article, we have proposed the Fedorov theorem as a technique to solve differential equations which describe the dynamics of open quantum systems. The method was applied to specific types of three-level and four-level systems. The generators studied in the article are in line with evolution models considered within laser Physics. Thus, the results provide valuable insight into the dynamics of relaxation systems. Various characteristics of disspative systems, such as the purity or the von Neumann entropy, can be investigated in the time domain based on the Fedorov theorem.

In the future, the Fedorov theorem shall be applied to other multi-level quantum systems, which may bring significant advancement in understanding the dynamics of dissipative systems composed of atoms interacting with light. When a high-dimensional Hilbert space is concerned, we expect that by partial commutativity one can study closed-form solutions of evolution equations within the admissible subset of initial quantum states. Further research into the Fedorov theorem seems relevant for pure mathematics as well as in the context of physical applications.

\section*{Acknowledgments}

The author acknowledges financial support from the Foundation for Polish Science (FNP) (project First Team co-financed by the European Union under the European Regional Development Fund).

I would like to thank prof. Andrzej Jamiolkowski, who handed to me the original article written by Fedorov. I also thank dr. Takeo Kamizawa for his comments on partial commutativity.

\end{document}